\documentclass[5p]{elsarticle}

\usepackage{listings}
\usepackage{color}
\usepackage{soul}
\definecolor{myblue}{rgb}{0,0,0.9}
\definecolor{mygray}{rgb}{0.9,0.9,0.9}
\definecolor{mymauve}{rgb}{0.58,0,0.82}
\usepackage[cmex10]{amsmath}
\usepackage{graphicx}
\usepackage{amssymb}
\usepackage{stmaryrd}
\usepackage{algorithm,algorithmic}
\usepackage{amsthm}
\usepackage{wrapfig}

\newcommand{\ALOOP}[1]{\ALC@it\algorithmicloop\ #1%
  \begin{ALC@loop}}
\newcommand{\ENDALOOP}{\end{ALC@loop}\ALC@it\algorithmicendloop}

\newtheorem{definition}{\textbf{\emph{Definition}}}
\newtheorem{theorem}{\textbf{\emph{Theorem}}}

\definecolor{garrisonpink1}{rgb}{0.858, 0.188, 0.478}

\journal{Computers \& Security}

\begin{document}

\begin{frontmatter}

\title{SEDML: \underline{S}ecurely and \underline{E}fficiently Harnessing \underline{D}istributed Knowledge \\ in \underline{M}achine \underline{L}earning}

\author[rvt]{Yansong Gao}
\ead{yansong.gao@njust.edu.cn}

\author[rvt]{Qun Li}
\ead{120106222757@njust.edu.cn}

\author[focal]{Yifeng Zheng\corref{cor1}}
\ead{yifeng.zheng@hit.edu.cn}


\author[rvt]{Guohong Wang}
\ead{wgh@njust.edu.cn}

\author[rvt]{Jiannan Wei}
\ead{jnwei@njust.edu.cn}

\author[rvt]{Mang Su}
\ead{sumang@njust.edu.cn}

\cortext[cor1]{Corresponding author}
\address[rvt]{School of Computer Science and Engineering,
Nanjing University of Science and Technology, Nanjing, JiangSu, China}
\address[focal]{School of Computer Science and Technology,
Harbin Institute of Technology, Shenzhen, Guangdong 518055, China.}

\begin{abstract}
Training high-performing deep learning models require a rich amount of data which is usually distributed among multiple data sources in practice.
Simply centralizing these multi-sourced data for training would raise critical security and privacy concerns, and might be prohibited given the increasingly strict data regulations.
To resolve the tension between privacy and data utilization in distributed learning, a machine learning framework called private aggregation of teacher ensembles (PATE) has been recently proposed.
PATE harnesses the knowledge (label predictions for an unlabeled dataset) from distributed teacher models to train a student model, obviating access to distributed datasets.
Despite being enticing, PATE does not offer protection for the individual label predictions from teacher models, which still entails privacy risks.
In this paper, we propose SEDML, a new protocol which allows to securely and efficiently harness the distributed knowledge in machine learning.
SEDML builds on lightweight cryptography and provides strong protection for the individual label predictions, as well as differential privacy guarantees on the aggregation results.
Extensive evaluations show that while providing privacy protection, SEDML preserves the accuracy as in the plaintext baseline. Meanwhile, SEDML outperforms the state-of-the-art work by $43\times$ in computation and $1.23\times$ in communication.

\end{abstract}


\end{frontmatter}

\section{Introduction}

Deep learning (DL) models have been employed in various applications including medical diagnosis, speech recognition, face recognition, and financial fraud detection~\cite{lecun2015deep,bakator2018deep}, given their unprecedented performance. Training DL models for a high accuracy performance relies on rich data, which is usually collected from multiple data sources and aggregated in a centralized data center.
However, sharing the data directly to a single centralized party for training is not always possible because of severe privacy concerns, especially for sensitive data such as medical images and bank information. In addition, the data aggregator must pay great attention to the data regulations such as the General Data Protection Regulation (GDPR, effective from May,
2018)~\cite{gdpr}, California Privacy Rights Act (CPRA, effective from Jan,
2021)~\cite{cpra}, and China Data Security Law (CDSL, effective from Sep 2021) \cite{cdsl}. The aggregator could face severe legal issues whenever the collected data is misused or inappropriately processed. One solution is to training the DL model over centralized but encrypted data. One representative work is by Mohassel \textit{et al.} \cite{mohassel2017secureml}, which allows secure model training over encrypted data for protecting data privacy. However, the substantial amount of data also introduces major challenges to DL such as high data dimensionality and model scalability \cite{najafabadi2015deep,chen2014big}.
A more pragmatic solution is distributed learning that does not need to access local data, which can significantly reduce privacy leakages while still harnessing the distribut$\-$ed isolated rich data~\cite{gao2021evaluation}. 
One popular paradigm is Federated Learning (FL)~\cite{mcmahan2017communication}, where clients perform local training and only share model updates rather than raw data to the aggregator for updating a global model. 
Nonetheless, the model update parameters still expose notable information that can be exploited by an adversary to infer private client data through, e.g., membership inference attack \cite{rahman2018membership} and data inversion attack \cite{alves2019mlprivacyguard,khosravy2021model}.

To further reduce the information communicated with the aggregation server in distributed learning, Papernot \textit{et al.}~\cite{papernot2016semi} recently proposed a new machine learning framework, namely Private Aggregation of Teacher Ensembles, or PATE for short. Generally, the PATE framework harnesses the knowledge from distributed teacher models to train a student model~\cite{papernot2018scalable}. 
The teacher models are locally trained by clients over their private datasets and utilized to produce label predictions for an unlabeled training dataset queried by a requester.
The label predictions from multiple teacher models are aggregated by an aggregator or service provider, which are then returned to the requester for usage in training the student model.
Hence, the whole process avoids access to the clients' local datasets. 
To mitigate the potential leakages from the aggregated labels revealed to the requester, they also employ differential privacy and add calibrated noises in the aggregation process.
In this way, the PATE framework not only ensures the practicability of the model training, but also provides some privacy protection.

However, the direct expose of the individual label predictions to the aggregator could still leak private information about the teacher models or even the local datasets.

Given this, Xiang \textit{et al.}~\cite{XiangWWL20} recently proposed a design which leverages homomorphic encryption to support secure aggregation under the PATE framework, ensuring the confidentiality of individual label predictions.
Nevertheless, their design relies on expensive homomorphic encryption and suffers from expensive computation overheads.
Meanwhile, their design poses practical constraints on the clients which have to stay online for active participation (more detailed discussions can be found in Section \ref{sec:related}).

In light of the above, in this paper, we propose a new design for securely and efficiently harnessing the distribut$\-$ed knowledge in machine learning. Our design follows the machine learning paradigm in PATE for knowledge transfer, yet provides strong protection for the confidentiality of individual label predictions throughout the aggregation procedure, as well as ensures differential privacy guarantee on the aggregated labels. 
In comparison with the state-of-the-art \cite{XiangWWL20}, our design constructively takes advantage of lightweight additive secret sharing and promises much better practical efficiency.
We summarize our main contributions below:

\begin{itemize}
    \item  We propose SEDML, a new protocol for securely and efficiently harnessing the distributed knowledge in machine learning. SEDML builds on lightweight cryptography and ensures strong protection for individual label predictions during the secure aggregation procedure, and differential privacy on the aggregated labels.
    \item  We propose an efficient method in the secret sharing domain for secure identification of the highest (noisy) vote count during the secure aggregation procedure.
    This method is mainly based on secure extraction of the most significant bit in the secret sharing domain so as to allow efficient and secure comparison.

    \item We conduct a comprehensive performance evaluation on SEDML, in terms of computation, communication, and accuracy. Extensive results validate that the accuracy performance in SEDML is comparable to the plaintext baseline. Meanwhile, compared to the state-of-the-art security design \cite{XiangWWL20}, SEDML takes $43\times$ less computation time and $1.23\times$ less communication.

\end{itemize}

The rest of the paper is organized as follows. Section~\ref{sec:preliminary} introduces some preliminaries. 
Section \ref{sec:problem} gives the problem statement.
Section \ref{sec:sedml} presents the detailed design of SEDML.
Section \ref{sec:experiment} shows the experiment results.
Section \ref{sec:related} discusses the related work.
Section \ref{sec:conclusion} concludes this paper.

\section{Preliminaries}\label{sec:preliminary}
\subsection{Knowledge Transfer via Aggregating Teacher Ensembles}

The PATE framework proposed in \cite{papernot2018scalable} harnesses the knowledge from distributed teacher models to train the student model. 
This framework is advantageous in that the training of the student model does not need to access local datasets held by a set of clients.

There are three components in the framework: teacher model, aggregation mechanism, and student model.
Each teacher model is independently trained using a client's local dataset, which could be privacy-sensitive.

The student model is trained under the supervision of teacher models by distilling knowledge from all teacher models based on an aggregation mechanism.
In particular, it is assumed that the requester who wants to train the student model has access to a public but unlabeled dataset.
%
Each teacher model produces a label prediction for each sample in the dataset.
The label predictions on each example are then aggregated in plaintext domain through a dedicated mechanism with differential privacy guarantees, which produces a aggregated label.
The student model is then trained on those samples labeled through the aggregation mechanism.

\subsection{Additive Secret Sharing}\label{sec:secretsharing}

Our design will rely on a lightweight cryptographic technique, additive secret sharing, to achieve a secure and efficient realization for aggregating teacher ensembles.
In particular, we will make use of $2$-of-$2$ additive secret sharing.

Given a value $\alpha \in \mathbb{Z}_{2^l}$, its $2$-of-$2$ additive secret sharing is a pair $([\alpha]_0=\alpha-r$, $[\alpha]_1=r)$, where $r$ is a random value in $\mathbb{Z}_{2^l}$ and the subtraction is done in $\mathbb{Z}_{2^l}$ (i.e., result is modulo $2^l$).
Each share reveals no information about the original value $\alpha$.

Suppose that two values $\alpha$ and $\beta$ are secret-shared among two parties $\mathcal{P}_0$ and $\mathcal{P}_1$, i.e., $\mathcal{P}_0$ holds $[\alpha]_0$ and $[\beta]_0$ while $\mathcal{P}_1$ holds $[\alpha]_1$ and $[\beta]_1$.
The secret sharing $[\alpha\pm \beta]$ can be computed locally where each party $\mathcal{P}_i$ ($i\in \{0,1\}$) directly computes $[\alpha\pm\beta]_{i}=[\alpha]_{i}\pm [\beta]_{i}$.
Multiplication by a constant $\gamma$ on the value $\alpha$ can also be done locally, i.e., $[\alpha\cdot \gamma]_{i}=\gamma \cdot [\alpha]_{i}$. 
Multiplication over $[\alpha]$ and $[\beta]$ can be supported by using the Beaver's multiplication triple \cite{Beaver91a,Corrigan-GibbsB17}.
That is, given the secret sharing of a multiplication triple $(t_1,t_2,t_3)$ where $t_3=t_1\cdot t_2$, $[\alpha\cdot \beta]$ can be obtained with one round of interaction between the two parties.
In particular, each party $\mathcal{P}_i$ first computes $[e]_i=[\alpha]_i - [t_1]_i$ and $[f]_i=[\beta]_i - [t_2]_i$.
Then, $\mathcal{P}_i$ broadcasts $[e]_i$ and $[f]_i$, and recovers $e$ and $f$.
Lastly, $P_i$ computes $[\alpha \cdot \beta]_i=i \cdot e \times f+ [t_1]_i \times f + [t_2]_i \times e + [t_3]_i$.

\subsection{Differential Privacy}

Differential privacy \cite{DworkMNS06} is a rigorous privacy notion which, intuitively, ensures that the output of aggregate statistics computation over a database is insensitive to changes in any data record.

\begin{definition}
\noindent\textbf{(($\epsilon$,$\delta$)-differential privacy)} A randomized mechanism $\mathcal{M}$ with domain $\mathcal{D}$ and range $\mathcal{R}$ satisfies($\epsilon$,$\delta$)-differential privacy if for any two adjacent inputs D, D$'$ $\in$ $\mathcal{D}$ and for any output S $\subseteq$ $\mathcal{R}$ it holds that:
\begin{equation}
Pr[\mathcal{M}(D)\in S] \le e^{\epsilon}\cdot Pr[\mathcal{M}(D')\in S]+\delta
\end{equation}
\end{definition}

In the application of differential privacy to machine learning, adjacent inputs refer to two datasets that differ by one training sample.
The randomized mechanism $\mathcal{M}$ is a training algorithm.
The natural interpretation for the parameters $\epsilon$ and $\delta$ are as follows: $\epsilon$ represents the upper limit on the privacy loss, and $\delta$ represents the probability that the privacy guarantee may not hold.

R\'enyi Differential Privacy (RDP) \cite{Mironov17} generalizes pure differential privacy ($\delta=0$), with the following advantages.
Firstly, it has nice composition property. Secondly, it provides a cleaner way to capture the privacy guarantees of Gaussian noise used for $(\epsilon,\delta)$-differential privacy.   
The RDP mechanism is defined based on R\'enyi divergence, as stated below:

\begin{definition}
\noindent\textbf{(R\'enyi Divergence)}. The R\'enyi divergence between two distributions P and Q, with order  $\alpha$ ($\alpha$ $>$ 0 and $\alpha$ $\ne$ 1), is defined as:
\begin{equation}
D_\alpha(P \parallel Q) \triangleq \frac{1}{\alpha - 1} log \mathbb{E}_{x\backsim Q} \left[(\frac{P(x)}{Q(x)})^\alpha\right].  
\end{equation}
\end{definition}

\begin{definition}
\noindent\textbf{(R\'enyi Differential Privacy)}. A randomized mechanism $\mathcal{M}$ guarantees ($\alpha,\epsilon$)-RDP with $\alpha \ge$ 1 if for any neighboring datasets D and D$'$, 
\begin{small}
\begin{align}
D_\alpha(\mathcal{M}&(D)\!\parallel\! \mathcal{M}(D'))\!=\\ &\frac{1}{\alpha-1}log\mathbb{E}_{x\backsim \mathcal{M}(D)}\!\left[(\frac{Pr(\mathcal{M}(D)=x)}{Pr(\mathcal{M}(D')=x)})^{\alpha-1}\right]\le\!\epsilon.   
\end{align}
\end{small}
\end{definition}

\begin{theorem}
\label{thm:RDP-composition}
If a mechanism $\mathcal{M}$ consists of a sequence of adaptive mechanisms $\mathcal{M}_1$, . . . , $\mathcal{M}_k$ such that for any i $\in$ $[k]$, $\mathcal{M}_i$ guarantees($\alpha$, $\epsilon_i$)-RDP, then $\mathcal{M}$ guarantees ($\alpha$, $\sum_{i=1}^k$$\epsilon_i$)-RDP.
\label{th1}
\end{theorem}

\begin{theorem}
\label{thm:RDP-to-DP}
\noindent\textbf{(From RDP to DP)}. If a mechanism $\mathcal{M}$ guarantees ($\alpha$, $\epsilon$)-RDP, then $\mathcal{M}$ guarantees ($\epsilon +\frac{ log 1/\delta}{\alpha-1}, \delta$)-differential privacy for any $\delta$ $\in$ (0, 1).
\label{th2}
\end{theorem}

\section{Problem Statement}\label{sec:problem}

\subsection{System Architecture}

\begin{figure}[t!]
\centerline{\includegraphics[width=0.46\textwidth]{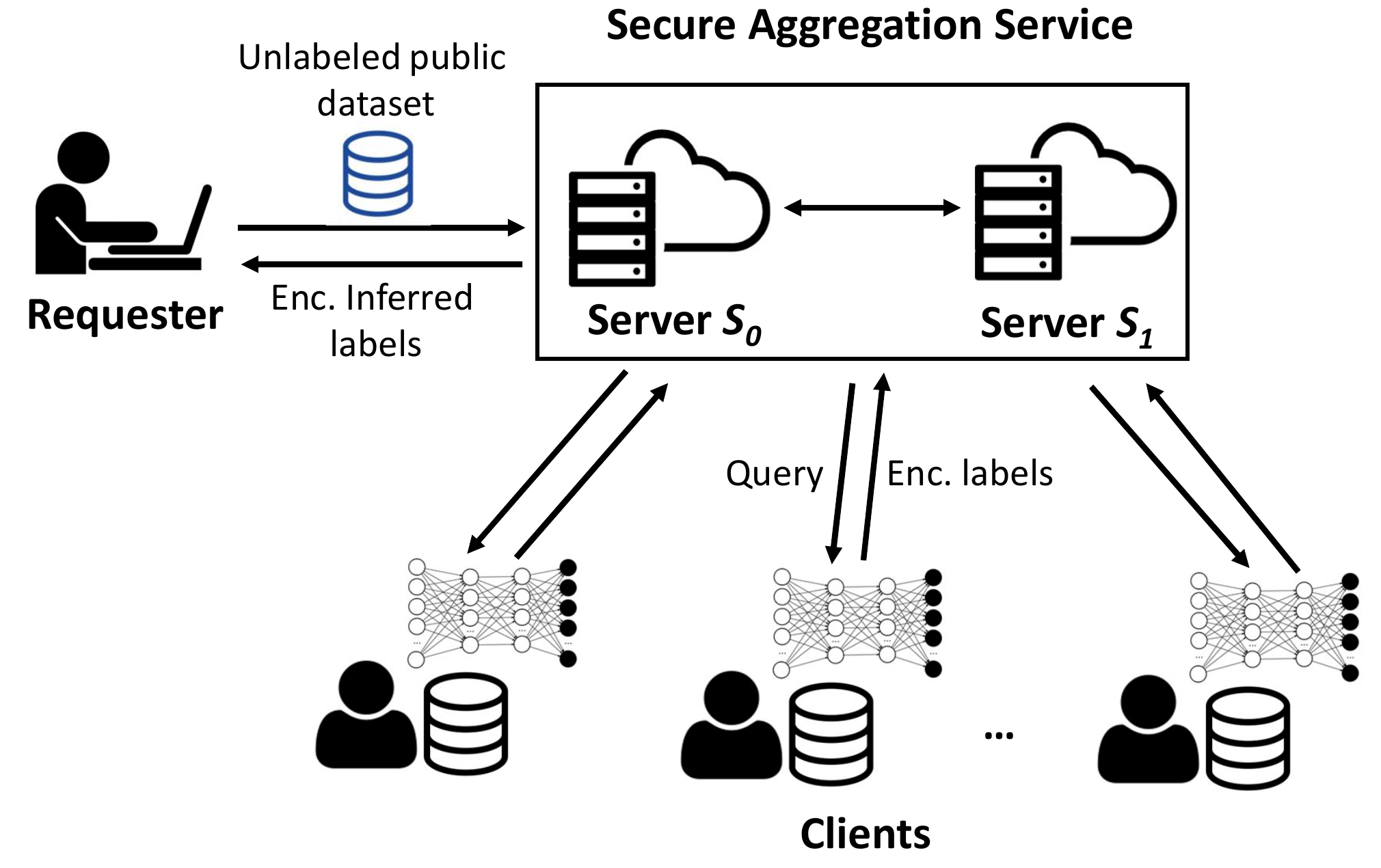}}
\caption{The system architecture.}
\label{fig:system_architecture}
\end{figure}

Fig. \ref{fig:system_architecture} illustrates the system architecture of SEDML that is aimed at securely and efficiently harnessing distributed knowledge in machine learning.
At the core, there are three parties: the requester, clients, and the secure aggregation service provider.
The requester wants to collect labels for a unlabeled public dataset via harnessing the collective knowledge of the clients, and then trains a model called student model.
Each client holds a proprietary model, namely teacher model, which is trained on private datasets locally.
On one hand, each client is interested in contributing knowledge to the training of the requester's student model via providing label predictions for the training examples in the public dataset through a teacher model trained over its local private dataset.
On the other hand, each client also has privacy concerns regarding the label predictions provided for the public dataset as they may reveal information about its teacher model and thus the private dataset on which the teacher model has been trained.
Hence, each client would only be willing to provide encrypted label predictions, and demand that security mechanisms should be put in place to safeguard their data privacy.

The secure aggregation service is a platform that bridg$\-$es the requester and the clients. It could be deployed on the cloud given the well-known advantages like scalability, ubiquitous access, and economical cost.
Similar to prior work \cite{XiangWWL20}, we consider that the secure aggregation service is jointly run by two cloud servers which are hosted by independent cloud providers.
We note that such a two-server model has recently gained increasing traction in both academic work \cite{RiaziWTS0K18,0002SKG19} and industrial sectors \cite{tfencrypted,crypten2020}.

In our system, the secure aggregation service receives label predictions in encrypted form from the clients, performs aggregation in the encrypted domain, and produces encrypted deferentially private aggregate label predictions for the public dataset, which are then returned to the requester on demand.

\subsection{Threat Assumptions and Security Goals}

In SEDML, we consider threats primarily come from the two cloud servers providing the knowledge aggregation service, under the commonly assumed semi-honest adversary model.
In particular, each cloud server will faithfully follows the protocol specifications of SEDML, yet may attempt to infer private sensitive information beyond their access rights, based on the messages received from the protocol execution.
Here, following the state-of-the-art \cite{XiangWWL20} as well as other works \cite{RiaziWTS0K18,0002SKG19,tfencrypted,crypten2020}, we assume the two cloud servers from different trust domains are non-colluding.
The rationale behind such non-collusion assumption is that cloud providers are business-driven parties and usually well-established companies, so they have least incentives to risk their reputations by acting maliciously.
With respect to the above threat model, our system aims to provide two following security guarantees: 

\begin{enumerate}

\item \textbf{Confidentiality for individual label predictions.} The label predictions from individual clients are kept confidential throughout the service flow.

\item \textbf{Differential privacy for individual clients.} The aggregated label predictions revealed to the requester should be differentially private so that inferring private information about individual clients from the aggregated label predictions is thwarted.

\end{enumerate}

\section{The Design of SEDML}\label{sec:sedml}

\subsection{Design Rationale}
To harness the distributed knowledge while being priva$\-$cy-friendly, SEDML is aimed at securely aggregating the label predictions collected from the teacher models held by a set of clients, so that a student model can be trained by the requester based on the training examples with aggregated labels.

We start with an overview of the aggregation mechanism (without considering differential privacy), which follows the plaintext-domain PATE framework \cite{papernot2018scalable} and the state-of-the-art security design \cite{XiangWWL20}.
Without loss of generality, we describe the aggregation of label predictions for one data sample $x$ in the unlabeled public dataset for the sake of simplicity.

Suppose there are $K$ clients, each of which holds a teacher model.
We use ${\bf y}_j$ ($j\in\{1,...,K\}$) to denote the label prediction from the $j$-th teacher model for the training example $x$ in the unlabeled public dataset. 

The label prediction $\mathbf{y}_j$ from teacher model $j$ is an $N$-dimensional binary vector, given that there are $N$ classes in total. 
If the predicted class is the $i$-th class, the $i$-th element in the vector $\mathbf{y}_j$ --- denoted by $\mathbf{y}_j(i)$ --- is 1, and all other elements are $0$. 

We denote the vote count for the $i$-th class as $n_i$, which is computed as $n_i=\sum\nolimits_j{\mathbf{y}_j(i)}$.
According to \cite{papernot2018scalable,XiangWWL20}, the aggregation of the label predictions, without considering differential privacy, works as follows.
Firstly, the vote count $n_i$ for each class $i$ is computed, followed by the computation of the highest vote count $n^*$, i.e., $n^*=\max(n_1,\cdots,n_N)$. 
The highest vote count $n^*$ is then compared to a threshold $T$.
If $n^*\ge T$, which means there is a consensus among the teacher models, the class $i^*$ corresponding to $n^*$ is output as the aggregated label for the training example $x$.
Otherwise, a termination symbol $\perp$ is returned.
Therefore, only the training examples with a consensus-reached aggregated class label will be used in training the student model.

\noindent\textbf{Challenges.}
Although the above aggregation mechanism has no direct access to clients' local models and datasets, the label predictions from the teacher models can still pose a great threat to data privacy \cite{li2020label,erdogan2021unsplit}, which is overlooked in the PATE framework. 
Therefore, the aggregation mechanism should be performed while keeping the label predictions collected from the teacher models confidential.
In addition, for the training examples which have a class whose highest votes are greater than the threshold, the aggregated label should be produced in the encrypted form as well and delivered to the requester on demand.
For other training examples, they should be discarded as no consensus is reached among the teacher models.

That is, it is expected that throughout the whole workflow, the aggregation service \textit{only learns whether there is a consensus among the teacher models given a training example in the dataset, and nothing beyond.}

From the above aggregation procedure, it is noted that the aggregation of the label predictions from the teacher models requires the atomic operations of addition and comparison.
For securing the aggregation process, one may considering the use of homomorphic encryption as taken by the state-of-the-art work~\cite{XiangWWL20}.
However, homomorphic encryption is expensive and incurs significant performance overheads.

\noindent\textbf{Our Approach.} In SEDML, to ensure security while ensuring high efficiency, we resort to the lightweight technique of additive secret sharing for data encryption and processing, in contrast to the expensive homomorphic encryption used in \cite{XiangWWL20}. 
Despite that we note that addition in the additive secret sharing domain can be directly supported, as shown in the preliminaries (Section.~\ref{sec:secretsharing}). There is a dearth of efficiently supporting secure comparison in the additive secret sharing domain.
Our observation is that secure comparison of two values $x$ and $y$ in the additive secret sharing domain can be realized via securely extracting the most significant bit
of the subtraction result $x-y$ between two values in the ring $\mathbb{Z}_{2^l}$ \cite{ZhengDW19,LiuZYY21}. 
We further observe that the MSB extraction can be ingeniously achieved via implementing a full adder logic in the secret sharing domain.

Inspired by prior work \cite{LiuZYY21}, we take advantage of the carry look-ahead adder for realizing secure and efficient comparison in SEDML, considering the fact that it obviates the cumbersome sequential carry computation and thus consumes much less number of rounds, in contrast to the standard ripple carry adder.
It follows two general steps as below.
\begin{enumerate}
    \item Firstly, a carry generate signal $G_i$ and a carry propagate signal $P_i$ are defined, which can be computed instantly based on the input bits $\{a_i\}$ and $\{b_i\}$, i.e., $G_i=a_i\cdot b_i$ and $P_i=a_i+b_i$.
    
\item Secondly, the carry bit computation can be formulated as $c_{i+1}=G_i+P_i\cdot c_i$.
Such formulation allows a carry to be computed without waiting for the carry to ripple through all previous phases.
Let us take a 4-bit carry look-ahead adder as an example.
We have $c_4=G_3+P_3\cdot c_3=G_3+P_3\cdot (G_2+P_2\cdot (G_1+P_1\cdot G_1))$.
\end{enumerate}

With the above formulation, the MSB of a secret $l$-bit value with shares $\{a_i\}^{l-1}_0$ and $\{b_i\}^{l-1}_0$ in bitwise form can be securely obtained via computing  $a_{l-1}+b_{l-1}+c_{l-1}$ in the secret sharing domain.
Hence, given such secure MSB extraction, we are able to achieve efficient secure comparison in the secret sharing domain, as opposed to the prior design \cite{XiangWWL20} that relies on the expensive homomorphic encryption.
However, there is another subtle challenge to be addressed specific to the secure label aggregation procedure.
In particular, during the comparison procedure, the relationship between the vote counts should not be revealed. 

To solve this issue, we propose to have the following efficient secure aggregation design.

Given the secret sharings $[a]$ and $[b]$ of two values $a$ and $b$, we first get the secret-shared comparison result $[e]$ underlying which the plaintext value $e$ is either $0$ ($a\ge b$) or $1$ ($a<b$).
Then, to achieve oblivious selection of the greater value, we compute the following: $[d] = [1-e]\cdot [a]+[e]\cdot [b]$, which is the secret sharing of the greater value. If $e$ = 1, $d$ = $b$; otherwise $e$ = $a$.
In this manner, we can hide the relationship among the vote counts throughout the comparison procedure in secure aggregation of individual label predictions.

Given the secret sharing $[n^*]$ of the greatest vote count $n^*$, we can then perform a secure comparison with the threshold $T$ to indicate whether there is a consensus among the teacher models.
To this end, the only information revealed to the cloud servers is whether there is consensus among the teacher models for each training example, and nothing beyond.
To ensure differential privacy on the aggregated labels, we further follow \cite{papernot2018scalable,XiangWWL20} and properly add Gaussian noises in the secure aggregation procedure. 

\subsection{The Secure Comparison Gadget} 


\begin{figure}[t!]
\centerline{\includegraphics[width=0.45\textwidth]{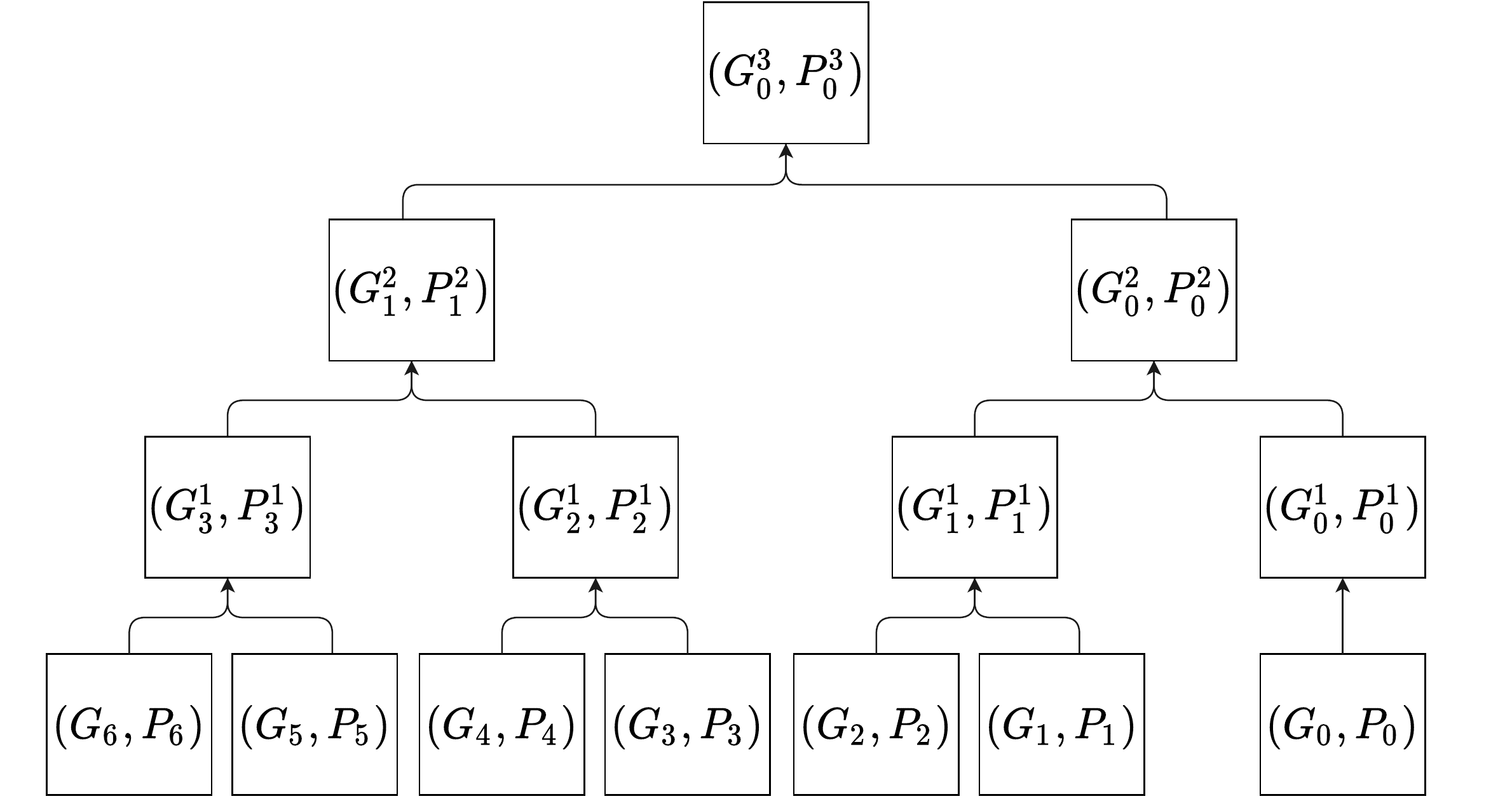}}
\caption{Example of carry calculation over 8-bit inputs with the carry look-ahead adder.}
\label{fig:carryCalculation}
\end{figure}

Before elaborating on the SEDML protocol, we introduce the secure comparison gadget based on MSB extraction in the secret sharing domain, as mentioned above.
The secure comparison gadget takes as input the secret sharings of two values $a$ and $b$, and outputs the secret sharing $[e]$ of the comparison result $e$.
We note that the whole computation procedure for the carry look-ahead adder can be organized in the form of a binary tree, where the bottom layer consists of the signals $G$ and $P$ corresponding to the input bits.
As an example, Fig. \ref{fig:carryCalculation} illustrates the computation for the case of an $8$-bit adder, where $G^3_0$ refers to the desired carry bit for the MSB computation.
Let us define an operator $\circ$ to be used during the computation.
As illustrated in Fig. \ref{fig:binaryOperator}, with $(G^*,P^*)=(G'',P'')\circ(G',P')$, we have $G^*=G''+G'P''$ and $P^*=P'P''$.
Let $\llbracket \cdot \rrbracket$ denote secret sharing in the ring $\mathbb{Z}_2$, as opposed to secret sharing $[\cdot]$ in the ring $\mathbb{Z}_{2^l}$.
Given the pre-generated multiplication triples in $\mathbb{Z}_2$ and $\mathbb{Z}_{2^l}$, the gadget $\mathsf{SCMP}_{ss}([a],[b])\rightarrow [e]$ proceeds as the follows:

\begin{figure}[t!]
\centerline{\includegraphics[width=0.25\textwidth]{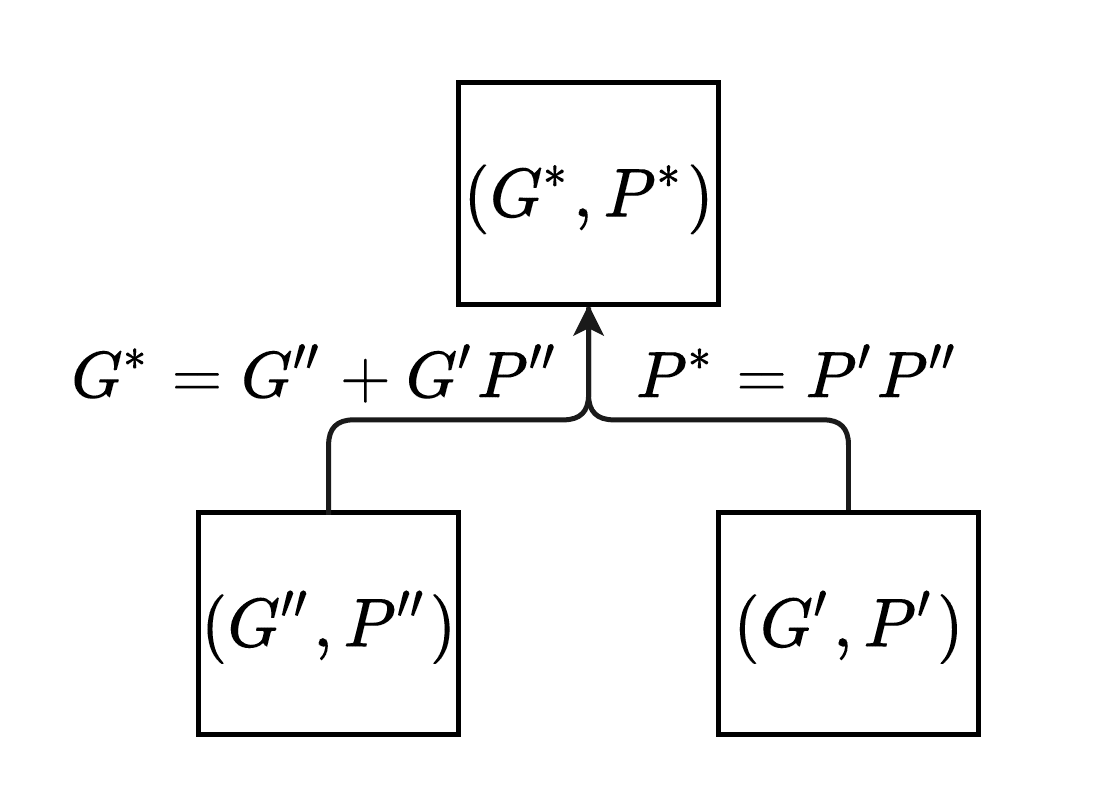}}
\caption{Binary operator for carry calculation.}
\label{fig:binaryOperator}
\end{figure}

\begin{enumerate}
\item Each cloud server $S_i$ computes $[f]=[a]-[b]$.
\item Let $x_{l-1},\cdots,x_0$ denote the bits for the share $[f]_0$ and $y_{l-1},\cdots,y_0$ for the bits of the share $[f]_1$. Also, for $j\in[0,l-1]$, $S_0$ sets $\llbracket x_j \rrbracket_0=x_j$ and $\llbracket y_j \rrbracket_0=0$; and $S_1$ sets $\llbracket x_j \rrbracket_1=0$ and $\llbracket y_j\rrbracket_1=y_j$. For $j\in[0,l-1]$, $S_0$ sets $\llbracket d_j\rrbracket_0=x_j$, and $S_1$ sets $\llbracket d_j\rrbracket_1=y_j$.

\item $S_0$ and $S_1$ compute $\llbracket G_j \rrbracket$=$\llbracket x_j \rrbracket$$\cdot$$\llbracket y_j \rrbracket$ and $\llbracket P_j \rrbracket$=$\llbracket x_j \rrbracket$$+$$\llbracket y_j \rrbracket$, for $j\in[0,l-1]$.

\item $S_0$ and $S_1$ set $(\llbracket G^1_0 \rrbracket,\llbracket P^1_0 \rrbracket)=(\llbracket G_0 \rrbracket,\llbracket P_0 \rrbracket)$.

\item $S_0$ and $S_1$ proceed through the following rounds to securely compute the MSB, i.e., the secure comparison result.

(a) In round $t=1$, for $k\in [1,l/2-1]$, $S_0$ and $S_1$ compute $(\llbracket G^1_k \rrbracket,\llbracket P^1_k \rrbracket)=(\llbracket G_{2k} \rrbracket,\llbracket P_{2k} \rrbracket)\circ (\llbracket G_{2k-1} \rrbracket,\llbracket P_{2k-1} \rrbracket)$. 
    
(b) In each round $t\in [2,\log l-1]$, for $k\in[0,l/2^t-1]$, $S_0$ and $S_1$ compute $(\llbracket G^t_k \rrbracket,\llbracket P^t_k \rrbracket)=(\llbracket G^{t-1}_{2k+1} \rrbracket,\llbracket P^{t-1}_{2k+1} \rrbracket)\circ (\llbracket G^{t-1}_{2k} \rrbracket,\llbracket P^{t-1}_{2k} \rrbracket)$. 

(c) In round $t=\log l$, $S_0$ and $S_1$ compute $\llbracket G^{t}_0 \rrbracket=\llbracket G^{t-1}_1 \rrbracket +   \llbracket G^{t-1}_0 \rrbracket  \cdot \llbracket P^{t-1}_1 \rrbracket=\llbracket c_{l-1} \rrbracket$.

(d) $S_0$ and $S_1$ compute $ \llbracket e \rrbracket=\llbracket d_{l-1} \rrbracket  +\llbracket c_{l-1} \rrbracket$.

    \item $S_0$ and $S_1$ convert $\llbracket e \rrbracket$ in $\mathbb{Z}_2$ to $\mathbb{Z}_{2^l}$ as follows. $S_0$ sets 
    $[p_1]_0=\llbracket e \rrbracket_0$ and $[p_2]_0=0$, and $S_1$ sets $[p_1]_1=0$ and $[p_2]_1=\llbracket e \rrbracket_1$. Then, $S_0$ and $S_1$ compute $[e]=[p_1]+[p_2]-2[p_1][p_2]$.
    
\end{enumerate}

From the above, we can see that the secure comparison gadget takes $O(\log l)$ communication rounds.
Meanwhile, the procedure is fully conducted in the secret sharing domain with efficient arithmetic operations, with secret-shar$\-$ed inputs and output.

\subsection{The SEDML Protocol}

We now present the complete SEDML protocol that allows to securely and efficiently harness distributed knowledge in machine learning. 
It builds on additive secret sharing to perform secure aggregation of the individual label predictions for training examples in the public dataset provided by the requester.
It also provides strong assurance of differential privacy for clients engaged in the service.
The complete SEDML protocol is shown in Algorithm \ref{alg:sedml-protocol}, which is introduced below.

Given a training example $x$, each client $j$ produces a label prediction encoded as a binary vector $\mathbf{y}_j$, as introduced above.
For privacy protection, client $j$ encrypts the vector $\mathbf{y}_j$ under additive secret sharing.
In particular, client $j$ generates a vector $\mathbf{r}$ of random values sampled from $\mathbb{Z}_{2^l}$, and generates the shares $[\mathbf{y}_j]_0=\mathbf{r}$ and $[\mathbf{y}_j]_1=\mathbf{y}_j-\mathbf{r}$ through element-wise computation in the ring $\mathbb{Z}_{2^l}$.
Client $j$ then sends the share $[\mathbf{y}_j]_0$ to cloud server $S_0$ and the share $[\mathbf{y}_j]_1$ to cloud server $S_1$ respectively.
Upon receiving the secret shares $[\mathbf{y}]$ of the label predictions from the clients for a training example, the cloud servers perform aggregation over the secret shares to produce an aggregated label for the training example if there is a consensus among the teacher models, or terminate on that example otherwise.

The secure aggregation procedure works as follows. Fir$\-$stly, leveraging the additive property of secret sharing, the cloud servers sum up the secret-shared label prediction vectors $\{[\mathbf{y}_j]\}^K_{j=1}$ and produce $[\mathbf{n}]= \sum\nolimits_j {[\mathbf{y_j}]}$, which corresponds to the secret sharing of the votes for the classes.
Then, the cloud servers need to obtain the encrypted highest vote count $n^*$ among the votes.
By invoking the secure comparison gadget, the cloud servers can securely compare a pair of elements $[\mathbf{n}(p)]$ and $[\mathbf{n}(q)]$ in the secret-shared vector $[\mathbf{n}]$.
That is, we have $\mathsf{SCMP}_{ss}([\mathbf{n}(p)],[\mathbf{n}(q)])\rightarrow [e]$, where $[e]$ indicates the comparison result.
Note that the secret-shared triples needed in the secure comparison gadget can be pre-generated offline and distributed to the two cloud servers by the requester.  
To obtain the secret sharing of the greater element, the cloud servers compute $[d] = [1-e]\cdot [\mathbf{n}(p)]+[e]\cdot [\mathbf{n}(q)]$.
Applying such secure comparison procedure over the votes in the secret-shared vector $\mathbf{n}$, the cloud servers can produce the secret-shared highest vote count $[n^*]$.

To check whether there is a consensus among the teach$\-$er models for the training example $x$, the cloud servers proceed as follows.
Firstly, the cloud server $S_0$ adds a Gaussian noise $g\leftarrow \mathcal{N}(0, \sigma^2_1)$ to its share $[n^*]_0$, which leads to that the cloud servers now hold the secret sharing of the noisy highest vote count, i.e., $[n^*+g]$.
Here, $\mathcal{N}(0, \sigma^2_1)$ means that the Gaussian distribution with mean 0 and variance $\sigma^2_1$. 
Note that addition of Gaussian noise is due to the demand for differential privacy.
Then, the cloud servers invoke the secure comparison gadget which takes as input the secret sharings of the noisy highest vote $[n^*+g]$ and the threshold $[T]$.
That is, we have $\mathsf{SCMP}_{ss}([n^*+g],[T])\rightarrow [t]$.
The cloud servers then reconstruct $t$ by exchanging the shares of $t$.
If $t=1$,  we have $n^*+g < T$, so there is no consensus among the teacher models and the cloud servers terminate on the training example x. In such case, the training example $x$ is discarded and will not be used by the requester when training the student model.
If $t=0$, we have $n^*+g \ge T$, so there is a consensus among the teacher models.

The cloud servers now proceed to produce the secret-shared aggregated label for $x$. Firstly, the cloud server $S_0$ adds a Gaussian noise to each element of the vector $\mathbf{n}$ in the secret sharing domain, producing a secret-shared noisy vector $\mathbf{m}$.
In particular, for each element $\mathbf{n}(i)$, the cloud server $S_0$ samples a noise $g_i$ from the Gaussian distribution $\mathcal{N}(0,\sigma^2_2)$ and computes $[\mathbf{m}(i)]_0=[\mathbf{n}(i)]_0+g_i$. The cloud server $S_1$ sets $[\mathbf{m}(i)]_1=[\mathbf{n}(i)]_1$.
In such way, the secret sharing of the vector of noisy vote counts is generated.
The cloud servers then invoke the secure comparison gadget over the vector $[\mathbf{m}]$.
Here, it is noted that in the end the cloud servers need to identify the index $i^*$ of the greatest value in the vector $\mathbf{m}$ after the secure comparison procedure.
Therefore, while securely comparing two elements $\mathbf{m}(p)$ and $\mathbf{m}(q)$ of the vector $\mathbf{m}$ in the secret sharing domain, the cloud servers generate the secret sharing of the index $s$ of the greater value among them.
In particular, given that $\mathsf{SCMP}_{ss}([\mathbf{m}(p)],[\mathbf{m}(q)])\rightarrow [z]$, the cloud servers compute $[s] = [1-z]\cdot p+[z]\cdot q$.
It is easy to see that if $\mathbf{m}(p)\ge \mathbf{m}(q)$, we have $z=0$, so $s=p$; and otherwise $s=q$.
Applying such comparison procedure, the cloud servers finally obtain the secret-shared index $[i^*]$ of the highest noisy vote count in $\mathbf{m}$, which corresponds to the aggregated label for the training example $x$.
This secret sharing can be delivered to the requester on demand, from which the requester can recover the aggregated label $i^*=[i^*]_0+[i^*]_1$ for the training example $x$ and use it in training the student model.

\noindent\textbf{Remarks.}
It is noted that directly comparing the vote counts sequentially requires $O(m)$ rounds of interactions among the cloud servers. Although this is already a linear increased complexity, it may still be a bottleneck, especially when the system runs in high-latency networks. 
To counter this, specific interaction reductions can be applied in the process of secure comparison of the elements of the secret vector for further efficacy optimization.
In particular, we can partition the vote counts into groups with a size $2$.
Then, secure comparison can performed for the two values within each group in parallel, meaning that the communication can be batched.
The (secret-shared) greater values from the secure comparison in each group form new groups for the next round of computation. 
In the end, the secret-shared greatest vote count is produced.

Our SEDML protocol fully runs in the secret sharing domain without heavy cryptography, as opposed to the state-of-the-art design \cite{XiangWWL20} that relies on expensive homomorphic encryption.
We also note that the design of \cite{XiangWWL20} needs multi-round communication among the clients and the cloud servers, while the clients in SEDML can just go offline after submitting their encrypted label predictions.
Furthermore, we note that the design of \cite{XiangWWL20} requires the number of participating clients to be determined and fixed in the beginning. 
All clients are required to participate subsequently, and their design will fail even if one of the clients fails to participate and submit ciphertexts.
Our SEDML protocol is free of such practical restriction.

\begin{algorithm}[!t]
\caption{The Proposed SEDML Protocol}
\label{alg:sedml-protocol}
\begin{algorithmic}[1]
\REQUIRE  Individual label prediction vectors $\{\mathbf{y}_j\}$.

\ENSURE The aggregated label $i^*$ if there is a consensus among the teacher models, or $\perp$ otherwise.

\underline{Client:} // \textit{Encrypt the prediction vector.}

\FOR{\textbf{each} client $j$}

 
 \STATE Generate a vector of random values $\mathbf{r}$ $\in$ $\mathbb{Z}_{2^l}$ and set the secret shares as $[\mathbf{y}_j]_0=\mathbf{r}$ and $[\mathbf{y}_j]_1=\mathbf{y}_j-\mathbf{r}$.

 \STATE Send the share $[\mathbf{y}_j]_0$ to cloud server $S_0$ and the share $[\mathbf{y}_j]_1$ to cloud server $S_1$ respectively. 
\ENDFOR

\underline{Cloud servers $S_0$ and $S_1$:} // \textit{Phase 1: Secure Highest Vote Identification}

\STATE Sum up $\{[\mathbf{y}_j]\}^K_{j=1}$ and produce $[\mathbf{n}]= \sum\nolimits_j {[\mathbf{y_j}]}$.
\STATE $[n^*]=[\mathbf{n}(0)]$.
  \FOR{($i=1; i<N; i++$)}
    \STATE $\mathsf{SCMP}_{ss}([n^*],[\mathbf{n}(i)])\rightarrow [e]$.
    \STATE $[n^*] = [1-e]\cdot [n^*]+[e]\cdot [\mathbf{n}(i)]$.
  \ENDFOR

\STATE Produce the secret-shared highest vote count $[n^*]$.

\underline{Cloud servers $S_0$ and $S_1$:}// \textit{Phase 2: Secure Threshold Check.} 

\STATE $S_0$ adds $g\leftarrow \mathcal{N}(0, \sigma^2_1)$ to its share $[n^*]_0$ and produces $[n^*+g]_0=[n^*]_0+g$.

\STATE $S_0$ sets $[n^*+g]_1=[n^*]_1$.

 \STATE 
Invoke $\mathsf{SCMP}_{ss}([n^*+g],[T])\rightarrow [t]$ and reconstruct $t$;

\IF{$t=1$} \STATE return $\perp$;

\ELSE \STATE go to next phase;
 
\ENDIF

\underline{Cloud servers $S_0$ and $S_1$:} // \textit{Phase 3: Secure Consensus Label Identification $[i^*]$.}

  \FOR{($i=0; i<N; i++$)}
    \STATE $S_0$ samples $g_i\leftarrow \mathcal{N}(0, \sigma^2_2)$.
    \STATE $S_0$ computes $[\mathbf{m}(i)]_0=[\mathbf{n}(i)]_0+g_i$.
    \STATE $S_1$ sets $[\mathbf{m}(i)]_1=[\mathbf{n}(i)]_1$.
  \ENDFOR

\STATE $[m^*]=[\mathbf{m}(0)]$.
\STATE $[s]$=$[0]$
\FOR{($i=1; i<N; i++$)}
 \STATE $\mathsf{SCMP}_{ss}([m^*],[\mathbf{m}(i)])\rightarrow [e]$.
      \STATE $[m^*] = [1-e]\cdot [m^*]+[e]\cdot [\mathbf{m}(i)]$.
          \STATE $[s] = [1-e]\cdot [s]+[e]\cdot [i]$.
\ENDFOR

\STATE Set $[i^*]=[s]$ and send $[i^*]$ to the requester upon request.

\underline{Requester:} // \textit{Reconstruct the aggregated label $i^*$.}
\STATE $i^*=[i^*]_0+[i^*]_1$.
\end{algorithmic}
\end{algorithm}

\subsection{Security Analysis}

The SEDML protocol provides assurance on confidentiality of the individual label predictions from clients, as well as differential privacy guarantees for clients.
In particular, throughout the secure aggregation procedure, the cloud servers only learn whether there is a consensus among the teacher models for a training example, without learning the individual label predictions.
Furthermore, the aggregation results provide differential privacy guarantees, which prevents information leakage by inference on the aggregation results.
As the confidentiality is ensured by the use of cryptographic techniques, we prove such guarantee following the standard simulation-based paradigm.
We start with giving the ideal functionality.

\begin{definition}

The ideal functionality $\mathcal{F}$ of securely harnessing distributed knowledge in SEDML is modeled as follows.
Given a training example $x$, each client $j$ provides a label prediction vector $\mathbf{y}_j$ to $\mathcal{F}$.
The requester and the two cloud servers input nothing to $\mathcal{F}$.
Upon receiving $\{\mathbf{y}_j\}^{K}_{j=1}$ from the clients, $\mathcal{F}$ conducts aggregation. If there is a consensus among the teacher models, $\mathcal{F}$ outputs an aggregated label to the requester. Otherwise, $\mathcal{F}$ returns nothing. 
\end{definition}

\begin{definition}
\label{def_security}
A protocol $\Pi$ securely realizes $\mathcal{F}$ if it provides the following guarantees. We require that a corrupted and semi-honest cloud server $S_i$ ($i\in\{0,1\}$) leans no information about individual label predictions and the aggregated label.
Formally, a PPT simulator should $\mathsf{Sim}_{S_i}$ should exist and generate a simulated view $\mathsf{View}_{\mathsf{Sim}_{S_i}}$ for $S_i$ such that  $\mathsf{View}_{\mathsf{Sim}_{S_i}}$ is indistinguishable to the view $\mathsf{View}^{\Pi}_{S_i}$ of $S_i$ in the real protocol execution, i.e., $\mathsf{View}^{\Pi}_{S_i} \mathop  \approx \limits^c \mathsf{View}_{\mathsf{Sim}_{S_i}}$.
\end{definition}

\begin{theorem}
Our SEDML protocol securely realizes the functionality $\mathcal{F}$ according to Definition \ref{def_security}, given that the two cloud servers are semi-honest adversaries and non-colluding.
\end{theorem}

\begin{proof}
According to our security definitions, we need to show the existence of a simulator for either of the cloud servers.
In the SEDML protocol, the roles of the two cloud servers are symmetric, so it is sufficient to show a simulator $\mathsf{Sim}_{S_0}$ for the cloud server $S_0$.
Recall that the cloud server $S_0$ receives secret shares of label predictions in the very beginning and then works over the secret shares throughout the whole secure aggregation procedure, with interactions with the other cloud server $S_1$.
Attributing to the security of additive secret sharing, the secret shares received by $S_0$ are uniformly random and can be easily simulated by the simulator $\mathsf{Sim}_{S_0}$ using random values.

During the computation of secure aggregation, the interactions among the cloud servers are to securely compare the (noisy) votes, based on the secure comparison gadget $\mathsf{SCMP}$.
According to the construction of $\mathsf{SCMP}$, it takes as input secret-shared values and outputs secret-shared values as well, and the inner processing is secure addition and secure multiplication based on standard Beaver's triples in the secret sharing domain.
Assume the simulator for the standard triple-based secure multiplication is $\mathsf{Sim}^B$.
The simulator $\mathsf{Sim}_{S_0}$ can invoke $\mathsf{Sim}^B$ on random values for each interaction with the cloud server $S_0$.
The security of Beaver's triple trick ensures that the view simulated by $\mathsf{Sim}^B$ is indistinguishable from the view of the cloud server $S_0$ in every secure multiplication in the real execution.
The simulator $\mathsf{Sim}_{S_0}$ combines in order the view simulated by $\mathsf{Sim}^B$ on every secure multiplication, which are then used as its simulated view for the secure comparison gadget.
Recall that during the computation, there is a secure comparison step where the cloud servers securely compare the secret-shard highest (noisy) vote with a threshold to see if there is a consensus among the teacher models, and the result is revealed to them, i.e., the comparison result $t$.
For this step, $\mathsf{Sim}_{S_0}$ adjusts the honest server's share of $t$ such that the recovered value is indeed the consensus-checking result $t$.
This concludes the simulation in our SEDML protocol.
\end{proof}

Following prior works \cite{papernot2018scalable,XiangWWL20}, our SEDML protocol also adds differential privacy noises so as to prevent information leakage from the aggregated results.
We have the following theorem regarding the differential privacy guarantee.

\begin{theorem}

The SEDML protocol provides $(\epsilon,\delta)$-differen\-tial privacy, where $\epsilon=\sqrt {2(9/\sigma^2_1 + 2/\sigma^2_2)\log 1/\delta}  + (9/2\sigma^2_1 + 1/\sigma^2_2)$ and $\delta\in (0,1)$.
\end{theorem}

\begin{proof}
As our SEDML protocol applies differential privacy in the same way as the prior work \cite{XiangWWL20}, the proof is similar to \cite{XiangWWL20}. So we only give the main points here and omit the details.
Specifically, the differential privacy mechanisms involved in the SEDML protocol consists of the sparse vector technique and the report-noisy-maximum technique, which correspond to steps 5 to 19 and steps 20 to 32 in Algorithm. \ref{alg:sedml-protocol} respectively.
According to \cite{XiangWWL20}, the sparse vector technique satisfies $(\alpha,9\alpha/2\sigma^2_1)$-RDP, and the report-noisy-maximum technique satisfies $(\alpha,\alpha/\sigma^2_2)$-RDP.
So given the composition property of RDP in Theorem \ref{thm:RDP-composition}, the SEDML protocol satisfies $(\alpha,9\alpha/2\sigma^2_1+\alpha/\sigma^2_2)$-RDP.
By Theorem \ref{thm:RDP-to-DP}, we have $(\epsilon,\delta)$-differential privacy for the SEDML protocol, where $\epsilon\ge  \sqrt {2(9/\sigma^2_1 + 2/\sigma^2_2)\log 1/\delta} + (9/2\sigma^2_1 + 1/\sigma^2_2)$ \cite{XiangWWL20}.
\end{proof}

\section{Experiments}\label{sec:experiment}
\subsection{Setup}

To validate the performance of our SEDML design, two popular datasets consisting of SVHN and MNIST are used for comprehensive experiments. All these two datasets have been used in closely-related works \cite{papernot2018scalable,XiangWWL20}.

The SVHN (street view house number) dateset contains images of house numbers as seen from Google Street View images~\cite{netzer2011reading}. Each image contains a set of Arabic numbers from `0' to `9'. Each colorful image sample has a size of $32 \times 32 \times 3$. The training set has 73,257 images, and the testing set contains 26,032 images and 531,131 additional images---all the training and additional images are used for training teacher models. In other words, we use 604,388 samples as teachers' training samples. These samples are evenly distributed according to the number of teachers to ensure that the training samples of each teacher model do not overlap.
A certain number of samples from the 26032 test samples are reserved as the unlabeled public dataset, and the rest is used as the test samples of the student model. The MNIST dataset consists of handwritten digital pictures. There are 10 categories of pictures, corresponding to 10 Arabic numbers from `0' to `9'~\cite{lecun1998gradient}. The numbers of training and testing image samples are 60,000 and 10,000, respectively. Each gray image sample has a size of $28\times 28 \times 1$. Similar to SVHN, we use 60000 samples as teachers' training samples, which are evenly distributed. A certain number of samples from the 10000 test samples are reserved as the unlabeled public dataset, and the rest is used to test the student model accuracy.
Our experiments use the same convolutional neural networks as in the prior work \cite{papernot2018scalable}.

The cryptographic operations in SEDML only work with integers.
However, the original voting counts will become floating-point numbers after Gaussian noises for differential privacy are added. Therefore, we need  to convert floating-point numbers into decimal integers when doing secure comparison in ciphertext domain. In our SEDML, the maximum length of the number to be compared is 32 bits, and the maximum number of votes is 250 (given 250 teacher models), and the decimal bit corresponding to 32 bits is $10^{9}$. Therefore, for a floating-point number, we multiply it by $10^{7}$ and truncate it. Only the integer part is taken.
Note that two numbers participating in the comparison will be multiplied with the same magnitude to retain the same expansion.
Our SEDML protocol is implemented in Python.
We run the experiments on a server with an AMD Ryzen 5 4600H CPU, 16GB RAM and the Windows 10 operating system.

\begin{table}[t!]
    \centering
    \caption{Computation Cost of Different Phases on the Cloud}
    \begin{tabular}{c|c}
    \hline
    Phase & 
 \begin{tabular}{@{}c@{}} Running time (s)\end{tabular}\\
    \hline
     Highest Vote Identification & 23.249
 \\ \hline
     Threshold Check & 2.5119
 \\ \hline
     Consensus Label Identification & 21.575
\\ \hline
    Overall& 47.3359
 \\
    \hline
    \end{tabular}
    
    \label{tab:COMPUTATIONAL COSTS}
\end{table}

\begin{table}[t!]
    \centering
    \caption{Communication Cost of Different Phases on the Cloud}
    \begin{tabular}{c|c}
    \hline
    Phase &  \begin{tabular}{@{}c@{}} Cost (KB)\end{tabular}\\
    \hline
      Highest Vote Identification & 29628 \\ \hline
     Threshold Check & 3292 \\ \hline
      Consensus Label Identification & 28201 \\ \hline
     Overall & 61121 \\ \hline
    \end{tabular}
    
    \label{tab:COMMUNICATION COSTS}
\end{table}

\subsection{Performance Evaluation}
We firstly stick with the SVHN dataset for comprehensive evaluations on both computation and communication performance.

\subsubsection{Computation Performance}

The computation cost of the secure aggregation procedure in SEDML consists of the following components in three phases.

\vspace{2pt}\noindent$\bullet$ {\bf Phase 1: Secure Highest Vote Identification:} (line 5 to 11). In the first phase, SEDML needs to securely compare the voting results for each pair of classes. We use the secure comparison algorithm to find the class with the highest vote. So we specifically evaluate the computation time of these operations. 

\vspace{2pt}\noindent$\bullet$ {\bf Phase 2: Secure Threshold Check:} (line 12 to 19). In the second phase, after the Gaussian noise with variance $\sigma_1^2$ is added to the highest vote, secure threshold check is utilized to determine whether the noisy maximum vote is greater than the pre-determined threshold. We evaluate the computation time for such secure threshold check.

\vspace{2pt}\noindent$\bullet$ {\bf Phase 3: Secure Consensus Label Identification:} (line 20 to 32). Once the threshold check is passed, we enter the third phase of SEDML. To be precise, Gaussian noise with variance $\sigma_2^2$ is added to the votes of all classes, and then the maximum value is determined with the secure comparison algorithm among all the votes. 

In our experiments, the number of samples used to query teacher models is 1000. In other words, the unlabeled dataset size held by the requester is 1000. To avoid variance in the result due to single round, we have performed 10 test runs and present the averaged cost.
Table.~\ref{tab:COMPUTATIONAL COSTS} reports the computation cost in each phase, for secure aggregation for 1000 samples. We can see that the running time of the secure highest vote identification phase is close to that of the secure consensus label identification phase.
The latter is a bit smaller due to the fact that some samples do not pass the secure threshold check and they will not be involved in the third phase. 

As for the running time of the secure threshold check phase, it is only 2.5119 seconds, as only one secure comparison operation is needed.

\subsubsection{Communication Performance}
Regarding the communication performance,  we examine and report the sizes of messages transmitted among the cloud servers, which are independent of computing platforms.
The results are summarized in Table. \ref{tab:COMMUNICATION COSTS}. 
The size of messages transmitted in the secure consensus label identification is again slightly lower than that of secure highest vote identification phase.
This is because some samples are filtered out in the intermediate secure threshold check phase.
Overall, the data size communicated in the first and third phase is about 9 times of that in the second phase.
Because for each sample, only one security comparison algorithm needs to be used in phase 2, while nine security comparisons are required in phase 1 and phase 3. Therefore, the sizes of messages transmitted in phase 1 and phase 3 is 9 times that in phase 2.

\subsubsection{Scalability}
We further investigate the scalability of the proposed SEDML from two aspects: the running time as a function of the number of samples, and the running time as a function of the number of classes. 
It should be noted that here we report the total running time of the whole secure aggregation procedure.

We first fix the number of classes to 10.
The left subfigure in Fig. \ref{fig:time} details the computational time when the number of samples varies from 1000 to 5000. We can see that the running time scales linearly with the number of samples, ranging from $47.336$s to $231.720$s. 

Then we fix the number of samples at $1000$ while varying the number of classes. We set the number of classes as 10, 20, 30, 40, and 50, respectively.
For this scalablity experiment, we use synthetic data as we are measuring the running time whose complexity depends on the number of classes.
The right subfigure in Fig. \ref{fig:time} shows the running time as the  number of classes varies. We can see that the running time also increases \textit{linearly} with the increase of the number of classes. 
This is because the total running time is dominated by the secure comparison operations. In the first and third phase, the complexity of required secure comparison operations is $O(n)$, while the second phase only needs a one-time secure comparison. For instance, when the number of classes is 10, 9 secure comparisons are required in the first phase; and when the number of classes is 20, 19 secure comparisons are required in the first phase. Therefore, the running time in SEDML as a function of the number of classes is also  approximately $O(n)$ that is desirably linear.

\subsection{Accuracy Evaluation}
We perform accuracy evaluation from the following aspects:

\vspace{2pt}\noindent$\bullet$ {\bf Label Accuracy:} It refers to the percentage of correctly labeled samples in the public dataset sent to clients.

\vspace{2pt}\noindent$\bullet$ {\bf Student Model Accuracy:} It is simply the testing accuracy of the student model trained with samples that receive consensus labels through the secure aggregation procedure.

As described in Algorithm \ref{alg:sedml-protocol}, a queried sample will be discarded if the teacher models cannot reach a consensus on the label through a threshold check.
We set the threshold to $60\%$ in the following experiments unless otherwise stated.
Apparently, the label accuracy is directly related to the threshold setting. If a high threshold or a low threshold is used, the label accuracy is expected to decrease. We will give a detailed description later.

\begin{figure*}[t!]
    \centering
    \includegraphics[width=1\textwidth]{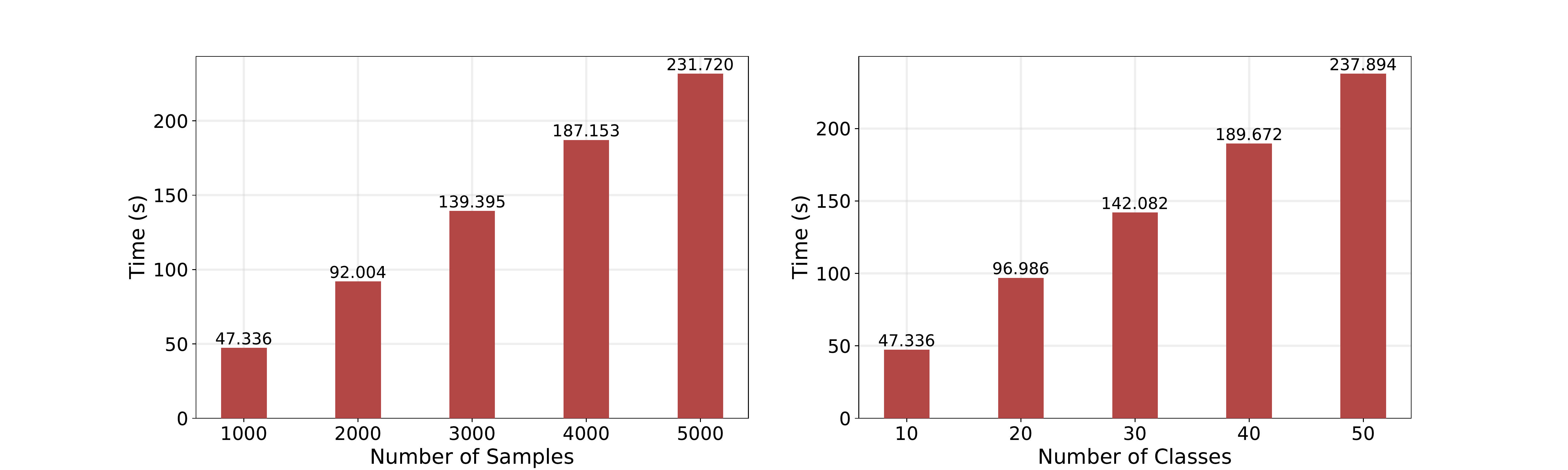}
    \caption{The scalability of SEDML with varying number of samples (left) and varying number of classes (right).}
    \label{fig:time}
\end{figure*}

\begin{figure*}[t!]
    \centering
    \includegraphics[width=1\textwidth]{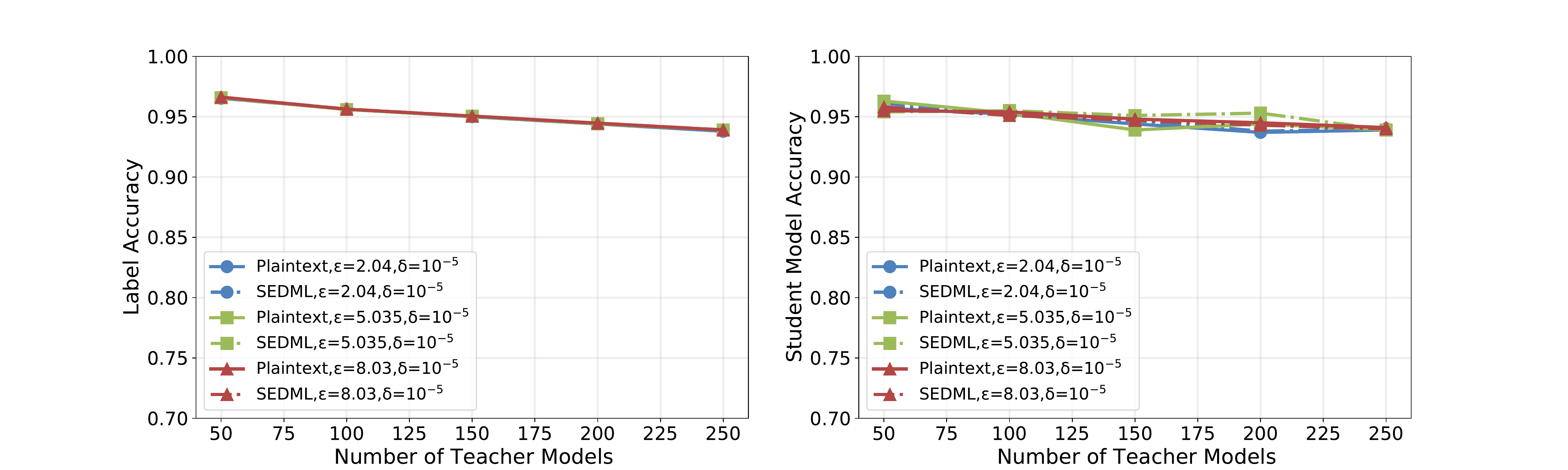}
    \caption{Accuracy evaluation results on MNIST. }
    \label{fig:mnist accuracy}
\end{figure*}

\begin{figure*}[t!]
    \centering
    \includegraphics[width=1\textwidth]{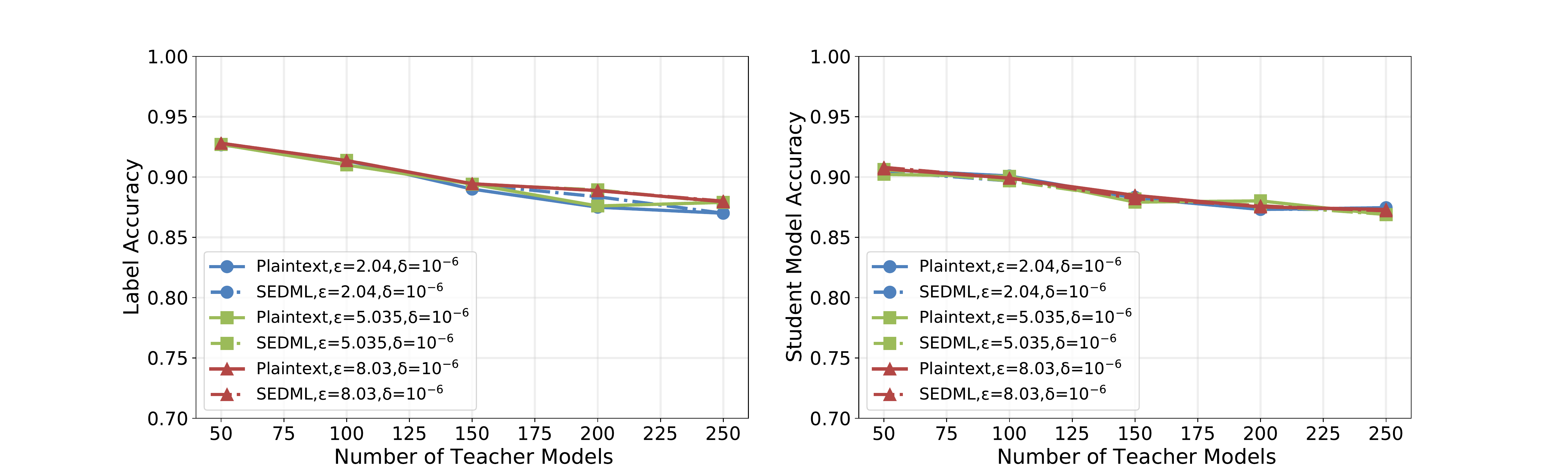}
    \caption{Accuracy evaluation results on SVHN. }
    \label{fig:svhn accuracy}
\end{figure*}

\begin{figure*}[t!]
    \centering
    \includegraphics[width=1\textwidth]{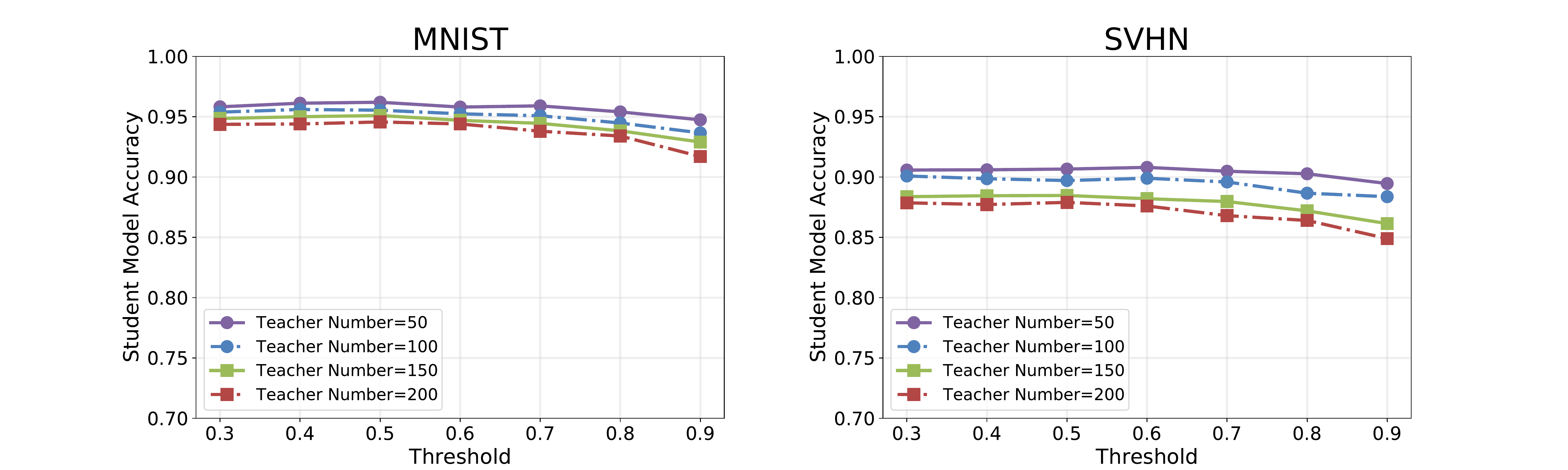}
    \caption{The student model accuracy with varying thresholds and number of teacher models. }
    \label{fig:threshold}
\end{figure*}

\begin{figure*}[t!]
    \centering
    \includegraphics[width=1\textwidth]{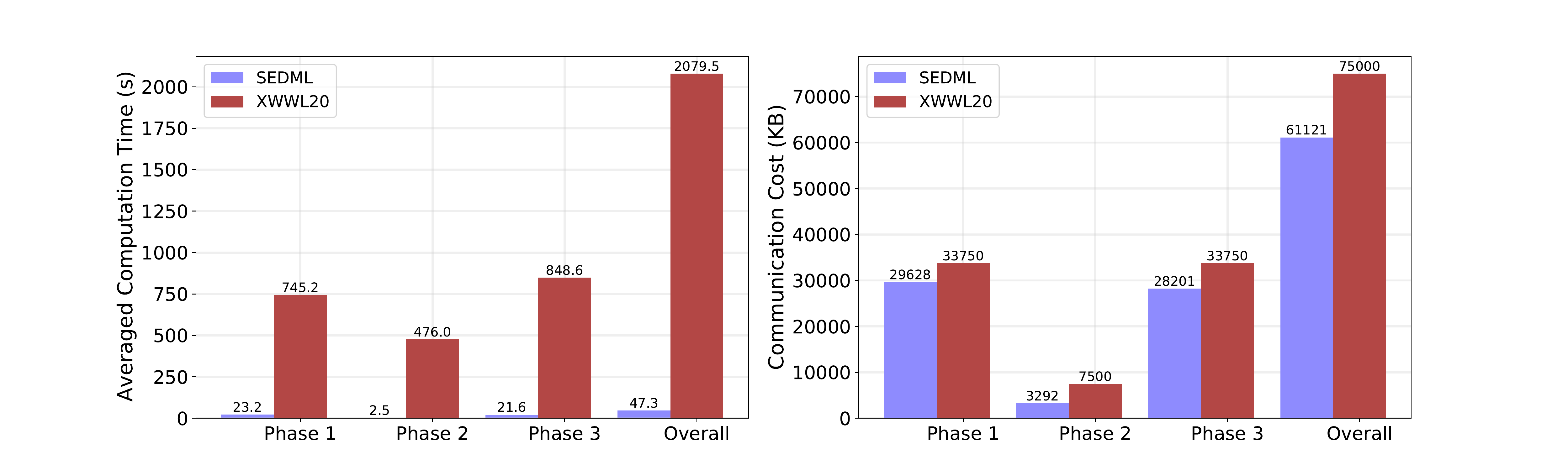}
    \caption{Performance comparison with the prior work XWWL20 \cite{XiangWWL20}. }
    \label{fig:comparison}
\end{figure*}

We compare the accuracy between SEDML and the plaintext baseline where the aggregation process is conducted on plaintext label prediction vectors.
Following \cite{XiangWWL20}, we evaluate different settings of the $\epsilon$ parameter: 2.04, 5.035, and 8.03.
For the SVHN and MNIST datasets, we set the $\delta$ to be $10^{-6}$ and $10^{-5}$ respectively. Because the scale of the SVHN we used is $10^6$ and the scale of the MNIST is $10^{5}$.

Fig. \ref{fig:mnist accuracy} displays the evaluation results on the label accuracy and student model accuracy over the MNIST dataset. The accuracy evaluation results over the SVHN dataset are given in Fig. \ref{fig:svhn accuracy}.
Both the accuracy of MNIST and SVHN will decrease as the number of teachers increases. This is because the increase of the number of teachers will reduce the samples that each client can use for training the teacher model---the total number of training samples in our experiments is fixed. Therefore, the accuracy of the teacher model will be reduced, which in turn will also affect the accuracy of the student model.

It can be seen from the results that the accuracy gap between SEDML and the plaintext baseline is negligible. 
This demonstrates that SEDML promises efficiency in preserving privacy while retaining a comparable student model accuracy. 
This is because the secure comparison is only added in the ciphertext, and the other steps are consistent with the plaintext, so the accuracy will not be affected.

We also investigate the influence of the threshold on the accuracy. The student model accuracy results for MNIST and SVHN are detailed in Fig. \ref{fig:threshold}, when varying the threshold from 0.3 to 0.9. The number of queried samples is fixed to 9000 and the number of classes is fixed to 10. The left subfigure in Fig. \ref{fig:threshold} shows the student model accuracy on MNIST, where $\epsilon$ = 8.02 and $\delta$ = 10$^{-5}$. 
The right subfigure in Fig. \ref{fig:threshold} shows the student model accuracy on SVHN, where $\epsilon$ = 8.02 and $\delta$ = 10$^{-6}$.
It is observed that the optimal threshold is between 0.5 and 0.6.
Specifically, for the MNIST dataset, the optimal threshold is about 0.5, while for the SVHN dataset, the optimal threshold is about 0.6.
If the threshold is too small, some wrong labels with the same number of votes as the ground truth label will be falsely regarded as the final consensus label. This adds noisy training samples when training the student model, degrading its accuracy. 
If the threshold is too high, the votes of some correct labels can not exceed the threshold and are eventually discarded. 
In such case, less useful training samples are involved when training the student model, which can also deteriorate the student model accuracy. Therefore, the best threshold lies between 0.5 and 0.6. This empirical observation aligns with the prior work \cite{XiangWWL20}. 
This also accounts for why we choose 0.6 as the default threshold in our experiments.

\subsection{Comparison with Prior Art}\label{sec:comparison}

We compare SEDML with the state-of-the-art privacy-preserving design by Xiang \textit{et al.}~\cite{XiangWWL20}, simply referred to as XWWL20 in the following presentation.
We first make comparison in the computation cost.
%
The left subfigure in Fig. \ref{fig:comparison} compares the running time between SEDML and XWWL20. All phases in SEDML are at least $32\times$ faster than XWWL20. This is because SEDML only builds on lightweight cryptography, as opposed to the expensive homomorphic encryption in XWWL20. 
The overall running time is $43\times$ lower than that of XWWL20.
We also compare the communication cost of SEDML with that of XWWL20.
The results are given in the right subfigure of Fig.~\ref{fig:comparison}.

Overall, our communication cost is $1.23 \times$ less than that of XWWL20.

\section{Related Work}\label{sec:related}

Papernot \textit{et al.} \cite{papernot2016semi} propose a knowledge transfer framework for deep learning which aggregates the label predictions from a teacher ensemble to train a student model.
Later, Papernot \textit{et al.} \cite{papernot2018scalable} present the formal PATE framework, which introduces new noisy aggregation mechanisms and greatly improve the accuracy upon the prior work  \cite{papernot2016semi}.
In \cite{zhang2020towards}, Zhang \textit{et al.} propose to leverage generative adversarial nets to combine advanced noisy label training mechanisms and the PATE framework to further improve accuracy.
Despite being appealing, the above works operate in the plaintext domain and do not provide confidentiality protection for the label predictions collected from the teacher models.

To counter the confidentiality issue in the PATE framework as aforementioned, Xiang \textit{et al.}~\cite{XiangWWL20} present a homomorphic encryption-based design which leverages homomorphic encryption to encrypt the individual label predictions from the teacher models and employ two non-colluding servers to conduct aggregation in the homomorphic ciphertext domain.
Their solution, however, is unsatisfactory due to the following limitations.

Firstly, their design requires the number of clients that will submit label predictions to be fixed in advance.
So if any client fails to submit the label predictions later, their design will not correctly work out.
Our SEDML design does not have such requirement on fixing the number of clients in advance.
Secondly, the interactions between each client and the servers are not one-off in their design.
In particular, after the secure threshold check for a training example, all the clients and the cloud servers need to interact again.
So the clients should keep staying online until the secure threshold check has been performed for all the queried samples.
If any client fails to participate in the second round of interaction, their design cannot correctly proceed again.
In contrast, the clients in SEDML can just go offline after sending their encrypted label predictions to the cloud servers.
Last not but least, their design relies on expensive homomorphic encryption and is much slower than our SEDML which only uses lightweight cryptographic techniques.

\section{ Conclusion}\label{sec:conclusion}
To securely and efficiently harness the rich distributed isolated data, we have proposed SEDML, a new protocol that leverages the knowledge from distributed teacher models to train a student model.
SEDML relies on the lightweight additive secret sharing to allow secure and efficient aggregation of the individual label predictions collected from the teacher models.
We have performed extensive experimental evaluations on two popular real-world datasets MNIST and SVHN.
The experiment results have demonstrated that the accuracy performance in SEDML is comparable to the plaintext baseline, and that SEDML greatly improves upon the state-of-the-art work in both computation and communication.

\section{Acknowledgment}
We acknowledge support from the National Natural Science Foundation of China (62002167, 61702268) and National Natural Science Foundation of JiangSu (BK20200461).

\bibliography{References}

\end{document}